\documentclass{llncs}

\usepackage{geometry}
\usepackage{todonotes}
\usepackage{lipsum}
\usepackage{amsfonts}
\usepackage{graphicx}
\usepackage{epstopdf}
\usepackage{algorithmic}

\usepackage{amsmath}

\ifpdf
  \DeclareGraphicsExtensions{.eps,.pdf,.png,.jpg}
\else
  \DeclareGraphicsExtensions{.eps}
\fi

\usepackage{amsopn}

\usepackage{tikz}
\usepackage{hhline}
\usepackage{subfig}
\usepackage{fix-cm}
\usepackage{colortbl}
\usepackage{pgfplots}
\usepackage{xcolor}

\usetikzlibrary{decorations.pathreplacing,calc,arrows,positioning}
\newcommand{\tikzmark}[2][-3pt]{\tikz[remember picture, overlay, baseline=-0.5ex]\node[#1](#2){};}
\tikzset{brace/.style={decorate, decoration={brace}}, brace mirrored/.style={decorate, decoration={brace,mirror}}}

\newcounter{arrow}
\newcounter{brace}
\newcommand{\drawbrace}[3][brace]{
  \refstepcounter{brace}
  \tikz[remember picture, overlay]\draw[#1] (#2.center)--(#3.center)node[pos=0.5, name=brace-\thebrace]{};
}
\newcommand{\annote}[3][]{
  \tikz[remember picture, overlay]\node[#1] at (#2) {#3};
}

\newcommand{\biFMTwoA}{
  \begingroup
  \renewcommand{\arraystretch}{1}
  \begin{tabular}{ | l |}
    \multicolumn{1}{c}{$\mathcal I$} \\
    \multicolumn{1}{c}{} \\[-2ex]
    \hline
    $\;\;\;\;\;\;\;\;\;\;\;\;\;\;\;$ \\
    \\[1ex]
    \\ \hline
    \cellcolor[gray]{0.5}\tikzmark[xshift=-10pt,yshift=1.5ex]{mark1}\tikzmark[xshift=-10pt,yshift=1.5ex]{mark2}\tikzmark[xshift=+72pt,yshift=1.4ex]{ii1} \\[2ex] \hline
    \tikzmark[xshift=-10pt,yshift=1.5ex]{mark3}\tikzmark[xshift=-10pt,yshift=1.5ex]{mark4}\cellcolor[gray]{0.8}\tikzmark[xshift=+50pt,yshift=1.4ex]{ii2} \\ \hline
    \tikzmark[xshift=-10pt,yshift=1.5ex]{mark5}\tikzmark[xshift=-10pt,yshift=1.5ex]{mark6}\tikzmark[xshift=+50pt,yshift=1.5ex]{ii3} \\ \hline
    \tikzmark[xshift=-10pt,yshift=1.5ex]{mark7}\tikzmark[xshift=-10pt,yshift=1.5ex]{mark8}\tikzmark[xshift=+72pt,yshift=1.4ex]{ii4} \\
    \\
    \\
    \hline
  \end{tabular}
  \endgroup
  \drawbrace[brace, thick]{ii2}{ii3}
  \drawbrace[brace, thick]{ii1}{ii4}
  \annote[right]{brace-1}{$Pc_j\,$}
  \annote[right]{brace-2}{$P\,$}

  \begin{tikzpicture}[overlay,remember picture]
    \draw[overlay,->,thin,shorten <=2pt] (mark1.west) -- (mark2.east) node[xshift=-15pt, yshift=1pt]{$a$};
    \draw[overlay,->,thin,shorten <=2pt] (mark3.west) -- (mark4.east) node[xshift=-15pt, yshift=1pt]{$a'$};
    \draw[overlay,->,thin,shorten <=2pt] (mark5.west) -- (mark6.east) node[xshift=-15pt, yshift=1pt]{$b'$};
    \draw[overlay,->,thin,shorten <=2pt] (mark7.west) -- (mark8.east) node[xshift=-15pt, yshift=1pt]{$b$};
  \end{tikzpicture}
}

\newcommand{\biFMTwoB}{
  \begin{tabular}{ | l |}
    \multicolumn{1}{c}{$\mathcal I^{rev}$} \\
    \multicolumn{1}{c}{} \\[-2ex]
    \hline
    $\;\;\;\;\;\;\;\;\;\;\;\;\;\;\;$ \\ \hline
    \cellcolor[gray]{0.5}\tikzmark[xshift=+50pt,yshift=1.4ex]{iii5} \\ \hline
    \tikzmark[xshift=+50pt,yshift=1.4ex]{iii6} \\[-1ex] \hline
    \cellcolor[gray]{0.5}\tikzmark[xshift=+50pt,yshift=1.4ex]{iii7} \\ \hline
    \tikzmark[xshift=+50pt,yshift=1.4ex]{iii8} \\[-1ex] \hline
    \cellcolor[gray]{0.8}\tikzmark[xshift=-10pt,yshift=1.5ex]{marc1}\tikzmark[xshift=-10pt,yshift=1.5ex]{marc2}\tikzmark[xshift=+50pt,yshift=1.4ex]{iii1} \\ \hline
    \tikzmark[xshift=-10pt,yshift=1.5ex]{marc3}\tikzmark[xshift=-10pt,yshift=1.5ex]{marc4}\tikzmark[xshift=+50pt,yshift=1.4ex]{iii2} \\ \hline
    \tikzmark[xshift=-10pt,yshift=1.5ex]{marc5}\tikzmark[xshift=-10pt,yshift=1.5ex]{marc6}\tikzmark[xshift=+50pt,yshift=1.5ex]{iii3} \\[6ex] \hline
    \tikzmark[xshift=-10pt,yshift=1.5ex]{marc7}\tikzmark[xshift=-10pt,yshift=1.5ex]{marc8}\tikzmark[xshift=+50pt,yshift=1.4ex]{iii4} \\
    \hline
  \end{tabular}

  \drawbrace[brace, thick]{iii3}{iii4}
  \annote[right]{brace-3}{$P^{rev}\,$}
  \drawbrace[brace, thick]{iii1}{iii2}
  \annote[right]{brace-4}{$c_{j\phantom{-0}}P^{rev}\,$}
  \drawbrace[brace, thick]{iii5}{iii6}
  \annote[right]{brace-5}{$c_{j-2}P^{rev}\,$}
  \drawbrace[brace, thick]{iii7}{iii8}
  \annote[right]{brace-6}{$c_{j-1}P^{rev}\,$}
  \begin{tikzpicture}[overlay,remember picture]
    \draw[overlay,->,thin,shorten <=2pt] (marc1.west) -- (marc2.east) node[xshift=-20pt, yshift=3.5pt]{$a'_{rev}$};
    \draw[overlay,->,thin,shorten <=2pt] (marc3.west) -- (marc4.east) node[xshift=-20pt, yshift=0.5pt]{$b'_{rev}$};
    \draw[overlay,->,thin,shorten <=2pt] (marc5.west) -- (marc6.east) node[xshift=-20pt, yshift=0pt]{$a_{rev}$};
    \draw[overlay,->,thin,shorten <=2pt] (marc7.west) -- (marc8.east) node[xshift=-20pt, yshift=0pt]{$b_{rev}$};
  \end{tikzpicture}
}

\newcommand{\rankdictGraphics}{
  \begin{tikzpicture}[xscale=0.047em,yscale=-0.05em]
    \tikzstyle{block lines}=[black,line width=.1em]
    \node[anchor=base east] at (-0.4,0.4) {$B$};
    \node[anchor=base east] at (-0.4,1.2) {blocks $M$};
    \node[anchor=base east] at (-0.4,2) {superblocks $M'$};
    \node[anchor=base west] at (20,0.4) {$\ldots$};
    \node[anchor=base west] at (20,1.2) {$\ldots$};
    \node[anchor=base west] at (20,2) {$\ldots$};
    \draw[black, densely dotted, step=0.5cm] (0,0) grid +(20,0.5);
    \begin{pgfscope}
        \clip (-.1,0.0) rectangle (20.1, 4);
        \draw[block lines, step=2cm](0,0.9) grid +(20,0.3);
        \draw[block lines] (0,1.2) rectangle +(20,0);
        \draw[block lines, step=8cm](0,1.7) grid +(20,0.3);
    \end{pgfscope}
    \draw[decorate,decoration={brace,mirror,amplitude=.5em}] (0,2.5) -- +(8,0) node at (4,2.8) [below]{$\ell^2$};
    \draw[decorate,decoration={brace,mirror,amplitude=.5em}] (8,2.5) -- +(2,0) node at (9,2.8) [below]{$\ell$};
  \end{tikzpicture}
}

\usepackage{hyperref}

\usepackage{tikz}
\usetikzlibrary{arrows,snakes,backgrounds,patterns,matrix,shapes,fit,calc,shadows,plotmarks}
\providecommand{\Oh}{\mathcal{O}}
\providecommand{\oh}{o}
\providecommand{\Index}{\mathcal{I}}
\providecommand{\EPR}{EPR}

\newcounter{nodecount}
\newcounter{nodecountB}
\newcommand\tabnodeE[1]{\tikz \node[draw=none,text height=1.5ex,text depth=.25ex] {#1};}
\newcommand\tabnodeM[1]{\addtocounter{nodecount}{1} \tikz \node[draw=black,text height=1.5ex,text depth=.25ex] (\arabic{nodecount}) {#1};}
\newcommand\tabnodeB[1]{\addtocounter{nodecountB}{1} \tikz \node[draw=none,text height=1.5ex,text depth=.25ex] (B\arabic{nodecountB}) {#1};}

\begin{document}
\title{EPR-dictionaries: A practical and fast data structure for constant time searches in unidirectional and bidirectional FM-indices} 
\author{
  Christopher Pockrandt\thanks{FU Berlin
    (\protect\email{christopher.pockrandt@fu-berlin.de}, \protect\url{http://reinert-lab.de}).}
  \and
  Marcel Ehrhardt\thanks{FU Berlin (\protect\email{marcel.ehrhardt@fu-berlin.de}).}
  \and
  Knut Reinert\thanks{FU Berlin (\protect\email{knut.reinert@fu-berlin.de})}
}
\institute{Department of Computer Science and Mathematics, Freie Universit\"at Berlin, Germany}

\pagestyle{plain}

\maketitle


\begin{abstract}
The unidirectional FM index was introduced by Ferragina and Manzini in 2000 and allows to search a pattern in the index in one direction. The bidirectional FM index (2FM) was introduced by Lam et al. in 2009. It allows to search for a pattern by extending an infix of the pattern arbitrarily to the left or right. The method of Lam et al. can conduct one step in time $\Oh(\sigma)$ while needing space $\Oh(\sigma \cdot n)$ using constant time rank queries on bit vectors. Schnattinger and colleagues improved this time to $\Oh(\log \sigma)$ while using $\Oh(\log \sigma \cdot n)$ bits of space for both, the FM and 2FM index. This is achieved by the use of binary wavelet trees.

In this paper we introduce a new, practical method for conducting an exact search in a uni- and bidirectional FM index in $\Oh(1)$ time per step while using $\Oh(\log \sigma \cdot n) + \oh(\log \sigma \cdot \sigma \cdot n)$ bits of space. This is done by replacing the binary wavelet tree by a new data structure, the \emph{Enhanced Prefixsum Rank dictionary} (\EPR-dictionary).

We implemented this method in the SeqAn C++ library and experimentally validated our theoretical results. In addition we compared our implementation with other freely available implementations of bidirectional indices and show that we are between $\approx 2.6-4.8$ times faster. This will have a large impact for many bioinformatics applications that rely on practical implementations of (2)FM indices e.g. for read mapping. To our knowledge this is the first implementation of a constant time method for a search step in 2FM indices.
\end{abstract}


\section{Introduction}

It is seldom that new data structures or algorithms have such a large practical impact as full text indices had for biological sequence analysis. The so-called next-generation sequencing (NGS) allows to produce billions of small DNA strings called \emph{reads}, usually of length 100-250. It is an invaluable technology for a multitude of applications in biomedicine. In many of these applications finding the positions of the DNA strings in a reference genome (i.e., a large string of length $\approx 10^7-10^{10}$) is the first fundamental step preceding downstream analyses. Finding the positions of the reads is commonly referred to as \emph{read mapping}.

Because of sequencing errors and genomic variations not all strings occur exactly in a reference genome. Therefore approximate occurrences must be considered and algorithms for approximate string matching tolerating mismatches, insertions, and deletions must be applied to solve the problem. 

This has triggered a plethora of work in the field to implement fast and accurate read mappers. Many of the popular programs like Bowtie2  \cite{Bowtie2}, BWA \cite{BWA}, BWA-Mem \cite{Li_Aligning_2013}, Masai \cite{Siragusa_Fast_2013}, Yara \cite{Yara2015}, and GEM \cite{Marco-Sola_The} use as their main data structure a version of the FM index \cite{Ferragina_Opportunistic_2000} that was introduced by Ferragina and Manzini in 2000. The FM index is based on the Burrows-Wheeler transform (BWT) \cite{Burrows94ablock-sorting} of the given text, i.e., the genomes at hand, and conceptually some lookup tables containing counts of characters in prefixes of the text. In its original form it allows to search exactly for a pattern in \emph{one direction} by matching the characters of the pattern with characters in the BWT~\cite{Burrows94ablock-sorting} (i.e., extending a suffix of the pattern character by character to the left).
 It was later extended to the 2FM index by Lam et al. \cite{Lam_High_2009} and Schnattinger et al. \cite{Schnattinger_Bidirectional_2012}. The 2FM index allows to search in both directions, that means we can extend an infix of a pattern arbitrarily to the left or to the right.
In order to reduce its space requirements, the count tables are in practice replaced by efficient bit vector data structures with rank support (see for example \cite{jacobson1988succinct}).
The search method of Lam et al. can conduct one search step in a 2FM index in time $\Oh(\sigma)$ while needing space $\Oh(\sigma \cdot n)$ using constant time rank queries on bit vectors. Schnattinger et al. improved this time to $\Oh(\log \sigma)$ while using $\Oh(\log \sigma \cdot n)$ bits of space for both, the FM and 2FM index. This is achieved by the use of binary wavelet trees introduced by Grossi et al. \cite{Grossi_High_2003}. In the last years several theoretical results appeared that improved on this. However, none of those has found a way into a practical implementation.
 
In this paper we introduce a new method for conducting an exact search in a uni- and bidirectional FM index that needs $\Oh(1)$ time per step while using $\Oh(\log \sigma \cdot n) + \oh(\log \sigma \cdot \sigma \cdot n)$ bits of space. This is done by replacing the binary wavelet tree by a new data structure, the \emph{Enhanced Prefixsum Rank dictionary} (\EPR-dictionary). To our knowledge
 this is the first implementation of a constant time method for 2FM indices.
 We will show, that the method outperforms other implementations by several factors 
 at the expense of a slight increase in memory usage resulting in a very practical method.

In the following paragraph we will review the concepts of the FM and 2FM index as well as constant time rank queries very shortly (readers unfamiliar with this can find a more detailed description in the appendix).

\subsection{Introduction to the FM and 2FM Index}

Given a text $T$ of length $n$ over an ordered, finite alphabet $\Sigma = \{c_1, \dots, c_{\sigma}\}$ with $\forall \, 1 \leq i < \sigma: c_{i} <_{lex} c_{i+1}$, let $T[i]$ denote the character at position $i$ in $T$, $\cdot$ the concatenation operator and $T[1..i]$ the prefix of $T$ up to the character at position $i$. $T^{rev}$ represents the reversed text. We assume that $T$ ends with a sentinel character $\$\notin \Sigma$ that does not occur in any other position in $T$ and is lexicographically smaller than any character in $\Sigma$. The FM index needs the Burrows-Wheeler transform (BWT) of $T$. The BWT is the concatenation of characters in the last column of all lexicographically sorted cyclic permutations of the string $T$. We will refer to the BWT as $L$.

In contrast to suffix trees or suffix arrays, where a prefix $P$ of a pattern is extended by characters to the right (referred to as forward search $P \rightarrow Pc$ for $c \in \Sigma$), the FM index can only be searched using backward searches, i.e., extending a suffix $P'$ by characters to the left, $P' \rightarrow cP'$. Performing a single character backward search of $c$ in the FM index will require two pieces of information. First, $C(c)$, the number of characters in $L$ that are lexicographically smaller than $c$, second, $Occ(c, i)$, the number of $c$'s in $L[1..i]$. Given a range $[a, b]$ for $P$; i.e., the range in the sorted list of cyclic permutations that start with $P$, we can compute the range $[a^\prime, b^\prime]$ for $cP$ as follows: $[a^\prime, b^\prime]$ = $[C(c) + Occ(c, a - 1) + 1, C(c) + Occ(c, b)]$.

The 2FM index maintains two FM indices $\Index$ and $\Index^{rev}$, one for the original text $T$ and one for the reversed text $T^{rev}$. Searching a pattern left to right on the original text (i.e., conducting a forward search) corresponds to a backward search in $\Index^{rev}$; searching a pattern right to left in the original text corresponds to a backward search in  $\Index$. The difficulty is to keep both indices \emph{synchronized} whenever a search step is performed. W.l.o.g. we assume that we want to extend the pattern to the right, i.e., perform a forward search $P \rightarrow Pc_j$ for some character $c_j$. First, the backward search $P^{rev} \rightarrow c_jP^{rev}$ is carried out on $\Index^{rev}$ and its new range $[a^\prime_{rev}, b^\prime_{rev}] = [C(c_j) + Occ(c_j, a_{rev} - 1) + 1, C(c_j) + Occ(c_j, b_{rev})]$ is computed. The new range in $\Index$ can be calculated using the interval $[a,b]$ for $P$ in $\Index$ and the range size of the reversed texts index $[a^\prime, b^\prime] = [a+smaller,a+smaller+ b^\prime_{rev} - a^\prime_{rev}]$. To compute $smaller$, e.g. Lam et al. \cite{Lam_High_2009} propose to perform $\Oh(\sigma)$ backward searches $P^{rev} \rightarrow c_iP^{rev}$ for all $1 \leq i < j$ and sum up the range sizes, i.e., $smaller = \sum_{1 \leq i < j} Occ(c_i, b_{rev}) - \sum_{1 \leq i < j} Occ(c_i, a_{rev} - 1)$ leading to a total running time of $\Oh(\sigma)$.

The implementation of the occurrence table $Occ$ is usually \emph{not} done by storing explicitly the values of the entire table. Instead of storing the entire $Occ : \Sigma \times \{1,\dots,n\} \rightarrow \{1,\dots,n\}$ one uses the more space-efficient \emph{constant time rank dictionary}: for every $c \in \Sigma$ a bit vector $B_c[1..n]$ is constructed such that $B_c[i] = 1$ if and only if $L[i] = c$. Thus the occurrence value equals the number of $1$'s in $B_c[1..i]$, i.e., $Occ(c, i) = rank(B_c, i)$. Jacobson \cite{jacobson1988succinct} showed that rank queries can be answered in constant time using only $\oh(n)$ additional space per bit vector by employing a sum of two count arrays (i.e., \emph{blocks} and \emph{superblocks}) and a final \emph{in-block} count. Since then many other constant time rank query data structures have been proposed. For an overview we refer the reader to \cite{Navarro_Fast_2012} containing a comparison of various implementations. For readers unfamiliar with \emph{2-level rank dictionaries}, an explanation is given in the appendix.

\subsection{Recent improvements on the FM and 2FM index}
For large alphabets, it is not practical to maintain for each character a bit vector with rank support. In 2003 Grossi et al.~\cite{Grossi_High_2003} proposed the use of a more space efficient data structure for the FM index, called the \emph{(binary) wavelet tree} (WT) that was later used by Schnattinger \cite{Schnattinger_Bidirectional_2012} for an implementation of bidirectional FM indices. It is a binary tree of height $\Oh(\log \sigma)$ with a bit vector of length $n$ with rank support at each level. This reduces the space consumption by a factor of $\Oh(\frac{\log \sigma}{\sigma})$  in trade-off for an increased running time of $\Oh(\log \sigma)$. Schnattinger used the fact that not only the rank query for a given character $c$ can be computed in $\Oh(\log \sigma)$ but also the $smaller$ value can be computed in the same asymptotic time which is quite convenient for the 2FM index. Ferragina et al. proposed a new data structure in 2007~\cite{ferragina2007compressed}, the multi-ary wavelet tree, which could be used to speed up the needed rank queries of 2FM indices. In 2013 Belazzougui et al. proposed the first constant-time bidirectional FM index \cite{belazzougui2013versatile} using minimal perfect hashing, of which to our knowledge no implementation exists (see also \cite{Belazzougui_Optimal_2015} for an extended version). Our solution is based on bit vectors with rank support, which proved so far to be very fast in practice, in particular due to the {\tt popcount} machine operation.

\section{Theoretical results}
\label{sec:main}

In this section we present the main results of this paper. They are based on a simple observation and a new bit vector data structure with rank query support which allows us to improve upon the results of Lam and Schnattinger.
Our proposed method runs in constant time per step while using $\Oh(\log \sigma \cdot n) + \oh(\log \sigma \cdot \sigma \cdot n)$ bits of space for small alphabets (i.e., $\sigma < \log(n) / \log\log(n)$) which is in theory inferior in space consumption to the results of mentioned above (see \cite{belazzougui2013versatile}), but in practice very fast, and presents to our knowledge the first constant time implementation of 2FM indices with this space complexity.

Our first observation is simple. Instead of defining a bit vector for each $c \in \Sigma$ to map characters equal to $c$ in $L$ to $1$'s, we suggest using prefix sum bit vectors $PB_c$, i.e., $PB_c[i] = 1$ if and only if $L[i] \leq_{lex} c$ for $c \in \Sigma$. 

\begin{theorem}\label{thm:prefixsum}
A step in a bidirectional search can be performed in time $\Oh(1)$ using $\Oh(\sigma\cdot n)$ bits of space.
\end{theorem}

\begin{proof}
We define $\text{\it Prefix-Occ}(c_j, i)=rank(PB_{c_j},i)$; that means it counts the number of occurrences of a character lexicographically smaller or equal than $c_j$ up to position $i$. $\text{\it Prefix-Occ}(c_j, i)$ and thus the $smaller$ value for the 2FM index can now be computed by a single rank query $rank(PB_{c_j}, i)$, the original $Occ(c_j, i)$ value for backward searches needs only two rank queries and a subtraction, namely $Occ(c_j, i)=rank(PB_{c_j}, i) - rank(PB_{c_{j-1}}, i)$ (for the lexicographically smallest character $c_0$ no subtraction is necessary).
\end{proof}

Note that the bit vector for the lexicographically largest character can be omitted, since all bits will be set to $1$ and thus $rank(PB_{c_{\sigma}}, i) = i, \;\forall\, 1 \leq i \leq n$.

Our next idea is the main result of this work and will allow us to reduce the space complexity for both the FM  and the 2FM index while maintaining the optimal running time of $\Oh(1)$ per search step.
Instead of using normal bit vectors we use directly the binary encoding of the BWT (an idea already used by BWT-SW\cite{Lam_Compressed_2008}).
We call our data structure \emph{EPR-dictionary}, short for \emph{Enhanced Prefixsum Rank dictionary}.

\subsection{The EPR-dictionary}
\label{sec:eprdict}
The general idea of the EPR-dictionary is as follows.
Assuming an ordered alphabet $\Sigma=\{c_1,\ldots,c_{\sigma}\}$, each character $c_i$ is encoded by the binary value $ord_2(c_i)$ of its rank $i$.
Conceptually, we use the binary representation of the BWT to derive from it a \emph{spaced} bit vector representation for $PB_c$ for each $s\in\Sigma$.
Then we compute the auxiliary data structures (i.e., blocks and superblocks (see also appendix))
for each of those vectors. After we have those auxiliary structures which only need $o(n)$ bits space, we \emph{delete} the bit vectors and only retain the BWT. In practice the blocks and superblocks
are computed directly by a linear scan on the BWT. For the last in-block query, we show how to derive the counts $\text{\it Prefix-Occ}(c_j, i)$ from the BWT in constant time using a number of logical and arithmetic operations.

W.l.o.g. we assume that the block length is an even multiple of $\log \sigma$ to avoid case distinctions in the proof. In practice this holds since the in-block query is performed with popcounts on registers the length of which is a power of $2$. All bitmasks used for computing the in-block rank are exactly as long as a block.
For a character $c_j$ we define the \emph{rank bitmask} $rb(c_j)$ to be a binary sequence of concatenations of the pattern $0^{\lceil\log\sigma\rceil-1}1\cdot ord_2(c_j)$, i.e., $\lceil\log\sigma\rceil-1$ many bits set to $0$ followed by a $1$ followed by the binary encoding of the character $c_j$. For the DNA alphabet with $\Sigma = \{A, C, G, T\}$ and its binary encodings $\{00, 01, 10, 11\}$ the rank bitmask for $G \in \Sigma$ is for example $rb(G)=01\cdot 10\cdot 01\cdot 10\ldots$.

Counting the characters inside a block is done in two steps. The characters at even and odd positions are counted separately to generate space for an overflow bit. Therefore we need a bitmask $M_E$ masking characters at even positions from the bit vector. $M_E$ has $1$s for each even block of length $\lceil\log\sigma\rceil$, i.e., $M_E=00\cdot 11\cdot 00\cdot 11\cdot 00\cdot 11\ldots$. Finally, we need a bitmask $BM$ which filters out the lowest bit of each odd $\log\sigma$ block, i.e., $BM=01\cdot 00 \cdot 01\cdot 00\cdot 01\cdot 00\ldots$ for $\sigma=4$.

\emph{Step 1.} We first take the characters at odd positions inside the corresponding block of the BWT, subtract it from $rb(c_i)$, which will result in the rightmost bit of
 even character positions to be set to $1$ if and only if the character to the right is smaller or equal to $c_i$. We then obtain exactly those bits by masking with $BM$.

\[ B_{EPR}(c_i)_E = (rb(c_i) - (BWT\&M_E))\&BM \]

\emph{Step 2.} We then take the characters at even positions inside the corresponding block of the BWT by shifting them $\lceil\log\sigma\rceil$ bits to the right and masking with $M_E$. We can now continue as in step 1 by subtracting it from $rb(c_i)$, which will again result in a $1$ bit in the rightmost bit of
even character positions to be set to $1$ if and only if the character to the right is smaller or equal to $c_i$. We then apply the bitmask $BM$ to filter only these rightmost bits.

\[ B_{EPR}(c_i)_O = (rb(c_i) - ((\gg_{\lceil\log\sigma\rceil} BWT) \& M_E)) \& BM \]

Finally both bit vectors are merged with one of them shifted by $1$ to the left avoiding the rightmost bits to overlap. In practice this is faster than two popcount operations.

\[ B_{EPR}(c_i)=  B_{EPR}(c_i)_E | (\ll_{1} B_{EPR}(c_i)_O) \]

Since we used the binary encoding of the BWT, note that the underlying rank queries have to be adapted to $\text{\it Prefix-Occ}(c_j, i) = rank(B_{\text{\it EPR}}(c_j), (i-1) \cdot \left\lceil \log \sigma \right\rceil+1)$. It follows directly that $Occ(c_j, i)$ can by computed in constant time by observing that

\[
  Occ(c_j, i) = \begin{cases}
    \text{\it Prefix-Occ}(c_j, i) - \text{\it Prefix-Occ}(c_{j-1}, i) &\mbox{if } j > 0 \\
    \text{\it Prefix-Occ}(c_j, i) & \mbox{otherwise}
  \end{cases}
\]

\newcommand{\redBit}[1]{\begingroup\color{red}\textbf#1\endgroup}
\begin{figure}[h]
\centering
      \subfloat[step 1]{

\begin{tabular}{llcccccccc}
  $rb(G)$  & $\qquad$ & $01$  & $10$  & $01$  & $10$  & $01$  & $10$  & $01$  & $10$  \\
  $BWT \& M_E\phantom{\gg_{\lceil\log\sigma\rceil} ()}\quad$  & $-$  & $00$ & $01$ & $00$ & $01$ & $00$ & $11$ & $00$ & $11$  \\
      &  & - & (C) &  - & (C) & - & (T) & - &  (T) \\
  \hline
           & $=$  & $0\redBit{1}$ & $01$ & 0$\redBit{1}$ & $01$ & $0\redBit{0}$ & $11$ & $0\redBit{0}$ & $11$ \\
  BM  & $\&$ & $01$ & $00$ & $01$ & $00$ & $01$ & $00$ & $01$ & $00$ \\
  \hline
  $B_{\text{\it EPR}}(G)_E$ & $=$  & $0\redBit{1}$ & $00$ & $0\redBit{1}$ & $00$ & $0\redBit{0}$ & $00$ & $0\redBit{0}$ & $00$ \\
\end{tabular}

      }\qquad
      \subfloat[step 2]{

\begin{tabular}{llcccccccc}
  $rb(G)$  & $\qquad$ & $01$  & $10$  & $01$  & $10$  & $01$  & $10$  & $01$  & $10$  \\
  $(\gg_{\lceil\log\sigma\rceil} BWT) \& M_E\quad$  & $-$  & $00$ & $00$ & $00$ & $10$ & $00$ & $10$ & $00$ & $00$  \\
      &  & - & (A) & - & (G) & - & (G) & - & (A) \\
  \hline
           & $=$  & $0\redBit{1}$ & $10$ & $0\redBit{1}$ & $00$ & $0\redBit{1}$ & $00$ & $0\redBit{1}$ & $10$ \\
  BM  & $\&$ & $01$ & $00$ & $01$ & $00$ & $01$ & $00$ & $01$ & $00$ \\
  \hline
  $B_{\text{\it EPR}}(G)_O$ & $=$  & $0\redBit{1}$ & $00$ & $0\redBit{1}$ & $00$ & $0\redBit{1}$ & $00$ & $0\redBit{1}$ & $00$ \\
\end{tabular}

      }\qquad
      \subfloat[retrieving $B_{EPR}(G)$]{

       \begin{tabular}{llllllllll}
         $B_{\text{\it EPR}}(G)_E$ & $\qquad$ & $\phantom{.}01$ & $\phantom{.}00$ & $\phantom{.}01$ & $\phantom{.}00$ & $\phantom{.}00$ & $\phantom{.}00$ & $\phantom{.}00$ & $\phantom{.}00$ \\
         $\ll_{1} B_{\text{\it EPR}}(G)_O\qquad\qquad$ & $|$ & $\phantom{.}10$ & $\phantom{.}00$ & $\phantom{.}10$ & $\phantom{.}00$ & $\phantom{.}10$ & $\phantom{.}00$ & $\phantom{.}10$ & $\phantom{.}00$ \\
         \hline
         $B_{\text{\it EPR}}(G)$ & $=$ & $\phantom{.}11$ & $\phantom{.}00$ & $\phantom{.}11$ & $\phantom{.}00$ & $\phantom{.}10$ & $\phantom{.}00$ & $\phantom{.}10$ & $\phantom{.}00$ \\
         $popcount$ & $=$ & $\phantom{.}6$ & & & & & & & \\
       \end{tabular}
             }

\caption{An example for $\Sigma=\{A,C,G,T\}$ that shows how to perform an in-block rank query for characters smaller or equal to $G$ of the BWT substring $ACGCGTAT$. The resulting bit vector $B_{\text{\it EPR}}(G)$ has a $1$ for each character smaller or equal to $G$., i.e., all positions except those with a $T$.}\label{fig:prtransform}
\end{figure}

The EPR-transformed bit vector $B_{\text{\it EPR}}(c_j)$ is now a ''normal'' bit vector and thus allows us to compute the prefix sums for a string in constant time. This improves the running time of the 2FM index and makes it optimal in terms of speed.

Let us now take a look at the space consumption. For our exposition we define the block length of $\ell = \left\lfloor \frac{\log n}{2} \right\rfloor$ (if $\ell$ is not a multiple of $\lceil\log\sigma\rceil$ padding strategies can be applied).
Given a $B_{\text{\it EPR}}(c_j)$, for the $m$-th superblock we count the number of $1$'s (i.e., the number of occurrences of characters smaller or equal to $c_j$ in the corresponding BWT) from the beginning of $B_{\text{\it EPR}}$ to the end of the superblock in $M'[m][j]=rank(B_{\text{\it EPR}}(c_j), m\cdot \ell^2)$. As there are $\left\lfloor\frac{\lceil\log\sigma\rceil \cdot n}{\ell^2}\right\rfloor$ superblocks and $\sigma$ characters, $M'$ can be stored in
\[\Oh\left(\sigma\cdot\frac{\log\sigma\cdot n}{\ell^2} \cdot \log n\right)=\Oh\left(\sigma\cdot\log\sigma\cdot\frac{n}{\log n}\right)=\oh(\sigma \log \sigma \cdot n)\] bits.
For the $m$-th block we count the number of $1$'s from the beginning of the overlapping superblock to the end of the block in $M[m][j]=rank\big(B_{\text{\it EPR}}[1+k\ell..n](c_j), (m-k)\ell\big)$ where $k=\left\lfloor\frac{m-1}{\ell}\right\rfloor\ell$ is the number of blocks left of the overlapping superblock. $M$ has $\left\lfloor\frac{\lceil\log \sigma\rceil \cdot n}{\ell}\right\rfloor$ entries for every character and can be stored in
\[\Oh\left(\sigma\cdot\frac{\log \sigma\cdot n}{\ell} \cdot \log \ell^2\right)=\Oh\left(\sigma\cdot\log \sigma \cdot\frac{ n\cdot\log\log n}{\log n}\right)=\oh(\sigma \log \sigma \cdot n)\] bits.

Let $P$ be a precomputed lookup table such that for each possible infix $V$ of a bit vector $B_{\text{\it EPR}}(c_j)$ of length $\ell$, $i\in\left[1..\left\lfloor\frac{\ell}{\log\sigma}\right\rfloor\right]$ and $c_j \in \Sigma$ holds $P[V][i]=rank(V, (i-1)\cdot\lceil\log\sigma\rceil + 1)$. There are $2^\ell \cdot \left\lfloor\frac{\ell}{\log\sigma}\right\rfloor$ entries of value at most $\left\lfloor\frac{\ell}{\log\sigma}\right\rfloor$ and thus can be stored in \[\Oh\left(2^\ell \cdot \frac{\ell}{\log\sigma} \cdot \log\frac{\ell}{\log\sigma}\right)=\Oh\left(2^\frac{\log n}{2}\cdot\log (n-\sigma) \cdot \log\log (n-\sigma)\right)= \]
\[\Oh\left(\sqrt{n}\cdot \log n \cdot \log\log n \right)=\oh(n)\]

  bits. Note that we do need this lookup table only \emph{once}, since the position and counting of the bits set to $1$ is the same for all characters.

Equivalent to Theorem 1, we do not need to store blocks and superblocks for $c_{\sigma}$ since $rank(B_{\text{\it EPR}}(c_{\sigma}), i) = i, \;\forall\, 1 \leq i \leq n$.

\begin{theorem}[Constant time prefix sum query]\label{thm:ctr}
One search step in an FM index or 2FM index can be performed in $\Oh(1)$ time using $O(\log \sigma \cdot n) + \oh(\log \sigma\cdot\sigma\cdot n)$ bits of space.
\end{theorem}

\begin{proof}
	The $BWT$ can be stored in $\Oh(\log \sigma \cdot n)$, all the blocks, superblocks, and lookup table $P$ in $\oh(\log \sigma \cdot \sigma \cdot n)$ bits. A prefix sum rank query requires one superblock and block lookup as well as a constant number of arithmetic and logical operations on the last block which run all in constant time.
\end{proof}

\section{Experimental results}
\label{sec:experiments}

In this Section we will conduct computational experiments to validate our theoretical findings and to compare our FM and 2FM indices to another available implementation.
All of our tests were conducted on Debian GNU/Linux 7.1 with Intel® Xeon® E5-2667V2 CPUs at fixed frequency of 3.3 GHz to prevent dynamic overclocking effects. All data was stored on tmpfs, a virtual file system in main memory to prevent loading data just on demand during the search and thus effecting the speed of the search by I/O operations.

In the first part of the experiments we will test FM and 2FM indices with our new data structure (EPR) in comparison to the wavelet tree (WT) implementation which was previously the generic standard implementation in SeqAn~\cite{SeqAn08}. Additionally we will run the same benchmarks for other available 2FM implementations, namely the bidirectional wavelet tree by Schnattinger et al.~\cite{Schnattinger_Bidirectional_2012} which we will call 2SCH and the balanced wavelet tree implementation with plain bit vectors and constant-time rank support in the SDSL~\cite{gbmp2014sea} which we will refer to as 2SDSL.

The 2BWT by Lam et al. \cite{Lam_Compressed_2008} is unfortunately not generic and only works for DNA alphabets. We also were not able to retrieve all hits when switching between forward and backward searches on the same pattern. Unfortunately we couldn't reach the authors and thus excluded 2BWT from our comparisons.

Another implementation that is worth mentioning is the affix array by Meyer et al.~\cite{Meyer_Structator_2011}. Even though the affix array implementation is generic, the construction algorithm did not terminate for alphabets other than DNA in a reasonable amount of time (several days). Unfortunately the affix array is not stand-alone but part of an application and does not provide a documented interface. Hence we were not able to include the affix array in our tests within a reasonable time frame. Meyer compares the running time of their affix array implementation with 2SCH and states that the affix array is faster by a factor of $1.26$ to $2$. From that we can conclude that our 2FM index implementation using the EPR-dictionary is expected to be faster than the affix array implementation (see below).

\subsection{Runtime and space consumption}
For the first benchmark we want to make a comparison with alphabets of different sizes to test the predicted independence from $\sigma$ for the EPR implementation. The alphabet sizes are inspired by bioinformatics applications and are of size 4 (DNA), 10 (reduced amino acid alphabet \emph{Murphy10}), 16 (IUPAC alphabet) and 27 (protein alphabet).

We first generated a text of length $10^8$ with a uniform distribution and sampled 1 million queries of length $50$ from this text. The search in the FM and 2FM indices will determine the number of occurrences of the sampled strings. Our sampling will ensure that the text occurs at least once and the stepwise search is never prematurely stopped. This ensures that we have $50$ million single steps in searches. The unidirectional FM indices perform backward searches while for 2FM indices we search the right half of the query first (using forward searches) and then extend the other half of the pattern to the left by backward searches.

In the following we will refer to WT and EPR as unidirectional FM indices and to 2WT and 2EPR as bidirectional FM indices, all part of the SeqAn library.

Table \ref{tab:runningtime} gives an overview of the running times of all FM and 2FM index implementations. It shows the absolute runtimes as well as the speedup factor relative to the unidirectional resp. bidirectional wavelet tree implementation.
WT, 2WT, 2SCH, 2SDSL are all based on wavelet trees. Our bidirectional wavelet tree implementation 2WT has a similar runtime compared to 2SDSL. It is slightly faster especially for small alphabets.

\begin{table}[h]
\begin{center}
\begin{tabular}{|l||r|r||r|r||r|r||r|r|}
\hline
 \multicolumn{1}{|l||}{}      & \multicolumn{2}{c}{DNA} & \multicolumn{2}{c}{Murphy10} & \multicolumn{2}{c}{IUPAC} & \multicolumn{2}{c|}{Protein} \\
Index & time & factor & Time & factor & Time & factor & Time & factor \\
\hline
WT		&	6.59s	&	1.00	&	16.97s	&	1.00	&	21.61s	&	1.00	&	26.87s	&	1.00 \\
\textbf{EPR}     &	\textbf{3.63s}	&	\textbf{1.82}	&	\textbf{5.35s}	&	\textbf{3.17}	&	\textbf{5.65s}	&	\textbf{3.83}	&	\textbf{6.20s}	&	\textbf{4.34} \\
\hline
2WT		&	9.32s	&	1.00	&	19.15s	&	1.00	&	23.44s	&	1.00	&	28.83s	&	1.00 \\
\textbf{2EPR}    &	\textbf{4.69s}	&	\textbf{1.99}	&	\textbf{5.78s}	&	\textbf{3.31}	&	\textbf{5.67s}	&	\textbf{4.13}	&	\textbf{6.21s}	&	\textbf{4.64} \\
2SDSL   &	12.21s	&	0.76	&	20.58s	&	0.93	&	24.43s	&	0.96	&	29.76s	&	0.97 \\
2SCH    &	14.08s	&	0.66	&	22.18s	&	0.86	&	26.11s	&	0.90	&	31.81s	&	0.91 \\
\hline
\end{tabular}
\end{center}
\caption{Runtimes of various implementations in seconds and their speedup factors with respect to the unidirectional wavelet tree.}
\label{tab:runningtime}
\end{table}

Compared to the wavelet tree implementations the EPR implementation is between 80\% (for DNA) and 330\% (Protein) faster for unidirectional indices and between 100\% (for DNA) and 360\% (Protein) faster for bidirectional indices.

Since we anticipate the application of 2EPR to bioinformatics applications, we also compared the runtime of all implementations using the complete human genome sequence. We again searched one million sampled strings of length 50 exactly as described above. The relative results were very similar to the ones in Table \ref{tab:runningtime}, indeed even slightly better. 2SCH crashed with this data set. 2SDSL was the slowest implementation (15.44s) followed by 2WT (1.6 times as fast) and by 2EPR (2.8 times as fast as 2SDSL) which was again the fastest implementation.

 Our experiments also show that we were indeed able to eliminate the $\log \sigma$ factor of wavelet trees in practice, as predicted by Theorem \ref{thm:ctr}. While the runtime for the WT implementations grows for larger alphabets with $\log \sigma$ the runtime of EPR and 2EPR increases only slightly for larger alphabets which can be explained by larger indices and therefore more cache misses. This can be seen in the following Figure in which we plot the runtime for EPR for different alphabets and the runtime of WT divided by $\log\sigma$. The resulting times develop very similarly.

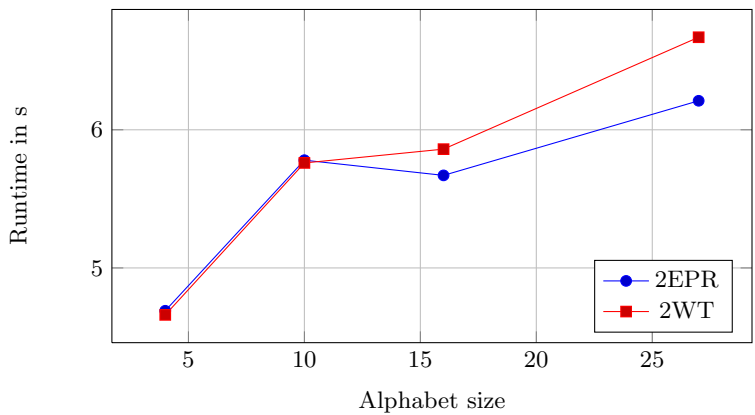
\begin{figure}[h]
	\begin{center}
        \begin{tikzpicture}
            \begin{axis}[height=6cm, width=10cm, grid=major, legend pos=south east,
                xlabel=Alphabet size, ylabel=Runtime in s]
                \addplot coordinates {
                    ( 4, 4.69)
                    (10, 5.78)
                    (16, 5.67)
                    (27, 6.21)
                };
                \addlegendentry{2EPR}
                \addplot coordinates {
                    ( 4, 4.66)
                    (10, 5.76)
                    (16, 5.86)
                    (27, 6.67)
                };
                \addlegendentry{2WT}
            \end{axis}
        \end{tikzpicture}
	\end{center}
	\caption{Plot of the runtime for EPR for different alphabets and the runtime of WT divided by $\log \sigma$.}
\end{figure}

When we compare the runtimes of the EPR and 2EPR, they behave as expected, i.e., the unidirectional index is slightly faster, since in each step of the bidirectional index we have to synchronize two indices.

All indices implemented in SeqAn (WT, EPR, 2WT, 2EPR) support up to 3 levels for rank dictionary support: blocks, superblocks and ultrablocks. The tests presented here were performed with a 2-level rank dictionary similar to the one explained in Section \ref{sec:eprdict} (or in the Appendix). Table \ref{tab:space} illustrates the practical space consumption for all previously discussed indices and of the affix array for DNA (larger alphabets did not finish within several days).

 Please note, that the other implementations may use versions of rank dictionaries different to the simple one explained in Section  \ref{sec:eprdict}. The numbers of FM indices given in Table \ref{tab:space} do neither account for storing the text itself nor for storing a compressed suffix array necessary to locate the matches in the text since the libraries use different implementations offering various space-time trade-offs. The running time of the backward and forward searches does not depend on it and the compressed suffix array implementation is independent from the used rank dictionary and thus interchangeable. A typical compressed suffix array implementation as used in the 2SDSL takes $\frac{n}{\eta}\log n$ (when sampling on the text instead of the suffix array). For a sampling rate of $10\%\phantom{..}(\eta = 10)$ the space consumption for our experiments would be $253$ MB and thus still much smaller than the affix array.

\begin{table}[h]
\begin{center}
\begin{tabular}{|l|c|c|c|c|}
\hline
Index & DNA	 & 	Murphy10    &   IUPAC   &   Protein \\
\hline
EPR	    & 	42	&	156	&	227	&	478 \\
2EPR    &   84	&	311	&	454	&	955 \\
WT	    & 	30	&	51	&	60	&	72 \\
2WT	    & 	60	&	102	&	120	&	144 \\
2SDSL   & 	68	&	105	&	122	&	145 \\
2SCH    &   75	&	108	&	123	&	146 \\
AF	 & 2670 & - & - & - \\
\hline
\end{tabular}
\end{center}
\caption{Space consumption of the rank data structure in Megabyte of various implementations}\label{tab:space}
\end{table}

The current implementation of the EPR and 2EPR in SeqAn interleaves the bit vector (i.e., the BWT) and precomputed block values but does not interleave superblock values. Reconsidering the design and storing block and superblock values close to the corresponding bit vector region could decrease the number of cache misses for one rank query to one cache miss and thus further improve the running time.

For larger alphabets one might also consider using a 3-level rank dictionary with smaller data types for blocks and superblocks which will reduce the space consumption noticeably at the expense of a slightly higher runtime (i.e., for the protein alphabet we reduce the space consumption from $955$ MB to $581$ MB while increasing the runtime from $6.21$ to $7.67$ seconds). The increased running time is due to another array lookup and thus still constant-time per step.

\subsection{Effect of the low order terms for space consumption}

In Table \ref{tab:expchangen} we show how quickly the $\oh(\log \sigma \cdot \sigma \cdot n)$ data structures for rank queries can be neglected for growing $n$.
For the WT and EPR implementations we measured the space needed for both the DNA and the IUPAC alphabet for $n=10^4,10^5,10^6,10^7,10^8,10^9$. We then divided the space consumption of both implementations by the factor in the $\Oh$-term, namely $\log \sigma\cdot n$.

For growing $n$ the $\Oh$-term should dominate the low order $\oh$-term, hence we would expect the resulting number converge to a constant. This is indeed true, as can be seen in Table \ref{tab:expchangen}. The EPR implementation converges faster than the WT, which is expected, since our $\oh$-term is larger than the one for the WT implementations. The effect of the $\oh$-terms falls for EPR from $10^5$ to $10^6$ by $35$ resp. $6$ percent, whereas the decline for WT is steeper with $84$ and $123$ percent. From size $10^7$ on, the low order terms are clearly dominated by the $\Oh$-terms.
\begin{table}[h]
\begin{center}
\begin{tabular}{|l|r|r|r|r|r|r|r|r|}
\hline
 Method ($\sigma$) & $10^4$ & $10^5$ & $10^6$ & $10^7$ & $10^8$ & $10^9$ \\
 \hline
   EPR (4)   &  2.4000   &  0.6000  &   0.4440   &   0.4292   &   0.4276   & 0.4274   \\
   EPR (16)  &  2.0000   &  1.2400   &  1.1680   &   1.1610  &  1.1602   & 1.1601    \\
   WT (4)     &  4.4000   &  0.6400   &  0.3480   &   0.3088  &  0.3056   & 0.3053   \\
   WT (16)    &  7.0000   &  0.8000   &  0.3580   &   0.3104  &  0.3057   & 0.3053    \\
   \hline  
\end{tabular}
\end{center}
\caption{Influence of the space consumption of the $\oh$-terms with increasing $n$.}\label{tab:expchangen}
\end{table}

\section{Conclusions}
\label{sec:conclusions}

In this paper we have introduced a new data structure, the EPR-dictionary, that enables constant time prefix sum computations for arbitrary, finite alphabets in $\Oh(\log\sigma\cdot n)+\oh(\log\sigma\cdot\sigma\cdot n)$ bits of space and works directly on the BWT. This allows two important data structures, the FM and 2FM index, to perform single search steps in time $\Oh(1)$.
We implemented the dictionary in the C++ library SeqAn and used it for an implementation of an FM and 2FM index. We compared its practical performance with the previous SeqAn implementation using wavelet trees and with other openly available implementations, among them the quasi standard for succinct data structures, the SDSL. We show that the EPR-dictionary implementation supports our theoretical claims, eliminates the $\log\sigma$ factor for searching in bidirectional indices, and performs between $80\%$ and $360\%$ faster than the wavelet tree implementation at the expense of a higher memory consumption. We compared our 2FM implementation against the available, open implementation of Schnattinger et al. (2SCH). Our implementation is between $3$ to $5.1$ times faster than 2SCH and $2.6$ to $4.8$ faster than the 2SDSL. We also showed that the additional space consumption is easily tolerable on normal hardware.

\section*{Acknowledgments}
We would like to acknowledge Enrico Siragusa for his previous implementations of the FM index in SeqAn. The first author acknowledges the support of the International Max-Planck Research School for Computational Biology and Scientific Computing (IMPRS-CBSC).
We also thank Veli M\"akinen and Simon Gog for very helpful remarks on a previous version of this manuscript during the Dagstuhl seminar 16351 ''Next Generation Sequencing - Algorithms, and Software For Biomedical Applications''.

\bibliographystyle{siamart_0216/siamplain}
\bibliography{references}

\newpage

\section*{Appendix}
In the appendix we give for the reader not familiar with FM and 2FM indices a short introduction.

\subsection*{Introduction to the FM and 2FM Index}

Given a text $T$ of length $n$ over an ordered, finite alphabet $\Sigma = \{c_1, \dots, c_{\sigma}\}$ with $\forall \, 1 \leq i < \sigma: c_{i} <_{lex} c_{i+1}$, let $T[i]$ denote the character at position $i$ in $T$, $\cdot$ the concatenation operator and $T[1..i]$ the prefix of $T$ up to the character at position $i$. $T^{rev}$ represents the reversed text. We assume that $T$ ends with a sentinel character $\$\notin \Sigma$ that does not occur in any other position in $T$ and is lexicographically smaller than any character in $\Sigma$. The FM index needs the Burrows-Wheeler transform (BWT) of $T$. The BWT is the concatenation of characters in the last column of all lexicographically sorted cyclic permutations of the string $T$ (see Figure \ref{fig:bwt} for an example). We will refer to the BWT as $L$.

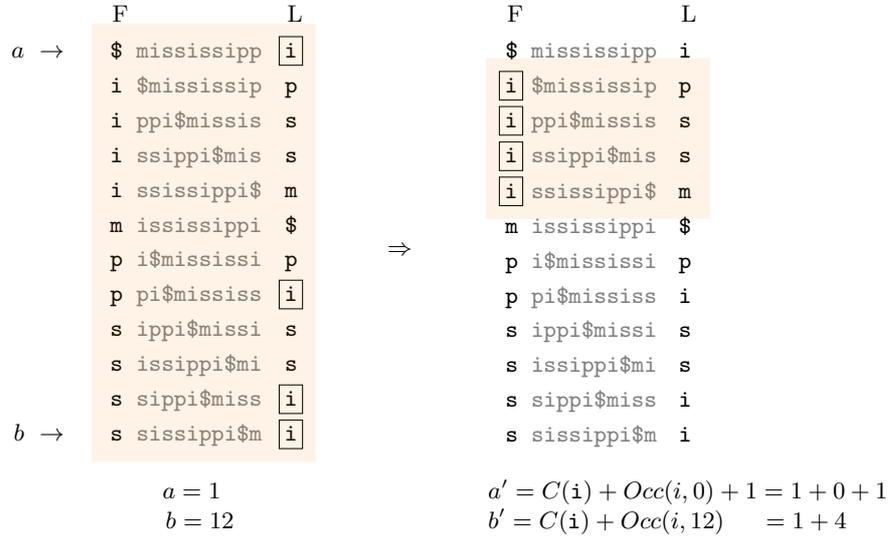
\begin{figure}[h]
   \begin{center}
    \hfill
    \small
	\tikzstyle{every picture}+=[remember picture,baseline]
	\tikzstyle{every node}+=[inner sep=2pt,outer sep=5pt]
	\bgroup
	\setlength{\tabcolsep}{0.05em}
\setcounter{nodecount}{0}
\setcounter{nodecountB}{0}
	\begin{tabular}[t]{rrlr}
	&F&&L\\
	\tabnodeE{\small$a\ \rightarrow\quad$}&\tabnodeB{\texttt{\$}}&\tabnodeE{\textcolor{gray}{\texttt{mississipp}}}&\tabnodeM{{\texttt{i}}}\\
	&\tabnodeE{\texttt{i}}&\tabnodeE{\textcolor{gray}{\texttt{\$mississip}}}&\tabnodeE{{\texttt{p}}}\\
	&\tabnodeE{\texttt{i}}&\tabnodeE{\textcolor{gray}{\texttt{ppi\$missis}}}&\tabnodeE{{\texttt{s}}}\\
	&\tabnodeE{\texttt{i}}&\tabnodeE{\textcolor{gray}{\texttt{ssippi\$mis}}}&\tabnodeE{{\texttt{s}}}\\
	&\tabnodeE{\texttt{i}}&\tabnodeE{\textcolor{gray}{\texttt{ssissippi\$}}}&\tabnodeE{{\texttt{m}}}\\
	&\tabnodeE{\texttt{m}}&\tabnodeE{\textcolor{gray}{\texttt{ississippi}}}&\tabnodeE{{\texttt{\$}}}\\
	&\tabnodeE{\texttt{p}}&\tabnodeE{\textcolor{gray}{\texttt{i\$mississi}}}&\tabnodeE{{\texttt{p}}}\\
	&\tabnodeE{\texttt{p}}&\tabnodeE{\textcolor{gray}{\texttt{pi\$mississ}}}&\tabnodeM{{\texttt{i}}}\\
	&\tabnodeE{\texttt{s}}&\tabnodeE{\textcolor{gray}{\texttt{ippi\$missi}}}&\tabnodeE{{\texttt{s}}}\\
	&\tabnodeE{\texttt{s}}&\tabnodeE{\textcolor{gray}{\texttt{issippi\$mi}}}&\tabnodeE{{\texttt{s}}}\\
	&\tabnodeE{\texttt{s}}&\tabnodeE{\textcolor{gray}{\texttt{sippi\$miss}}}&\tabnodeM{{\texttt{i}}}\\
	\tabnodeE{\small$b\ \rightarrow\quad$}& \tabnodeB{\texttt{s}}&\tabnodeE{\textcolor{gray}{\texttt{sissippi\$m}}}&\tabnodeM{{\texttt{i}}}\\
	\\
	&\multicolumn{3}{c}{$a=1\ \ $}\\
	&\multicolumn{3}{c}{$b=12$}\\
	\end{tabular}
	\egroup
	\begin{tikzpicture}[overlay]
	\filldraw[fill=orange,fill opacity=0.1,draw=none] (B1.north west)--(1.north east)--(4.south east)-|(B2.south west);
	\end{tikzpicture}
	\hfill
	\begin{array}[t]{c}
	\\\\\\\\\\\\\\\\\
	\ \Rightarrow
	\end{array}\hfill
	\bgroup
	\setlength{\tabcolsep}{0.05em}
	\setcounter{nodecount}{0}
	\setcounter{nodecountB}{0}
	\begin{tabular}[t]{lrlrl}
	\em&F&&L\\
	&\tabnodeE{\texttt{\$}}&\tabnodeE{\textcolor{gray}{\texttt{mississipp}}}&\tabnodeE{{\texttt{i}}}\\
	&\tabnodeM{\texttt{i}}&\tabnodeE{\textcolor{gray}{\texttt{\$mississip}}}&\tabnodeB{{\texttt{p}}}\\
	&\tabnodeM{\texttt{i}}&\tabnodeE{\textcolor{gray}{\texttt{ppi\$missis}}}&\tabnodeE{{\texttt{s}}}\\
	&\tabnodeM{\texttt{i}}&\tabnodeE{\textcolor{gray}{\texttt{ssippi\$mis}}}&\tabnodeE{{\texttt{s}}}\\
	&\tabnodeM{\texttt{i}}&\tabnodeE{\textcolor{gray}{\texttt{ssissippi\$}}}&\tabnodeB{{\texttt{m}}}\\
	&\tabnodeE{\texttt{m}}&\tabnodeE{\textcolor{gray}{\texttt{ississippi}}}&\tabnodeE{{\texttt{\$}}}\\
	&\tabnodeE{\texttt{p}}&\tabnodeE{\textcolor{gray}{\texttt{i\$mississi}}}&\tabnodeE{{\texttt{p}}}\\
	&\tabnodeE{\texttt{p}}&\tabnodeE{\textcolor{gray}{\texttt{pi\$mississ}}}&\tabnodeE{{\texttt{i}}}\\
	&\tabnodeE{\texttt{s}}&\tabnodeE{\textcolor{gray}{\texttt{ippi\$missi}}}&\tabnodeE{{\texttt{s}}}\\
	&\tabnodeE{\texttt{s}}&\tabnodeE{\textcolor{gray}{\texttt{issippi\$mi}}}&\tabnodeE{{\texttt{s}}}\\
	&\tabnodeE{\texttt{s}}&\tabnodeE{\textcolor{gray}{\texttt{sippi\$miss}}}&\tabnodeE{{\texttt{i}}}\\
	&\tabnodeE{\texttt{s}}&\tabnodeE{\textcolor{gray}{\texttt{sissippi\$m}}}&\tabnodeE{{\texttt{i}}}\\
	\\
	\multicolumn{5}{l}{$a'=C(\texttt{i})+Occ(i,0)+1=1+0+1$}\\
	\multicolumn{5}{l}{$b'=C(\texttt{i})+Occ(i,12)\;\;\;\;\;=1+4$}\\
	\end{tabular}
	\egroup
	\begin{tikzpicture}[overlay]
	\filldraw[fill=orange,fill opacity=0.1,draw=none] (1.north west)--(B1.north east)--(B2.south east)-|(2.south west);
	\end{tikzpicture}
	\hfill\ 
    \end{center}
  
    \caption{First step of the backwards search for $P=\texttt{ssi}$ in the FM-index for the text $T=\texttt{mississippi\$}$.
    The first interval $[a, b]$ is the whole range $[1,12]$. From all matrix rows we search those beginning with the last pattern character $P[3]=i$.
    From $Occ(i,1)=0$ and $Occ(i,12)=4$ follows $a'=C(i)+0+1=2$ and $b'=C(i)+4=5$. }\label{fig:bwt}
\end{figure}

In contrast to suffix trees or suffix arrays, where a prefix $P$ of a pattern is extended by characters to the right (referred to as forward search $P \rightarrow Pc$ for $c \in \Sigma$), the FM index can only be searched using backward search, i.e., extending a suffix $P'$ by characters to the left, $P' \rightarrow cP'$. Performing a single character backward search of $c$ in the FM index will require two pieces of information. First, $C(c)$, the number of characters in $L$ that are lexicographically smaller than $c$, second, $Occ(c, i)$, the number of $c$'s in $L[1..i]$. Given a range $[a, b]$ for $P$; i.e., the range in the sorted list of cyclic permutations that starts with $P$, we can compute the range $[a^\prime, b^\prime]$ for $cP$ as follows: $[a^\prime, b^\prime]$ = $[C(c) + Occ(c, a - 1) + 1, C(c) + Occ(c, b)]$. We will refer to the BWT together with tables $C$ and $Occ$ as \emph{FM index $\Index$} (see Figure \ref{fig:bwt} for an example of one search step).

The 2FM index maintains two FM indices $\Index$ and $\Index^{rev}$, one for the original text $T$ and one for the reversed text $T^{rev}$. Searching a pattern left to right on the original text (i.e., conducting a forward search) corresponds to a backward search in $\Index^{rev}$; searching a pattern right to left in the original text corresponds to a backward search in  $\Index$. The difficulty is to keep both indices \emph{synchronized} whenever a search step is performed. W.l.o.g. we assume that we want to extend the pattern to the right, i.e., perform a forward search $P \rightarrow Pc_j$ for some character $c_j$. First, the backward search $P^{rev} \rightarrow c_jP^{rev}$ is carried out on $\Index^{rev}$ and its new range $[a^\prime_{rev}, b^\prime_{rev}] = [C(c_j) + Occ(c_j, a_{rev} - 1) + 1, C(c_j) + Occ(c_j, b_{rev})]$ is computed. The new range in $\Index$ can be calculated using the interval $[a,b]$ for $P$ in $\Index$ and the range size of the reversed texts index $[a^\prime, b^\prime] = [a+smaller,a+smaller+ b^\prime_{rev} - a^\prime_{rev}]$. To compute $smaller$, Lam et al. \cite{Lam_High_2009} propose to perform $\Oh(\sigma)$ backward searches $P^{rev} \rightarrow c_iP^{rev}$ for all $1 \leq i < j$ and sum up the range sizes, i.e., $smaller = \sum_{1 \leq i < j} Occ(c_i, b_{rev}) - \sum_{1 \leq i < j} Occ(c_i, a_{rev} - 1)$ leading to a total running time of $\Oh(\sigma)$ (See Figure \ref{fig:2fmsearch} for an illustration).

\begin{figure}[htbp]
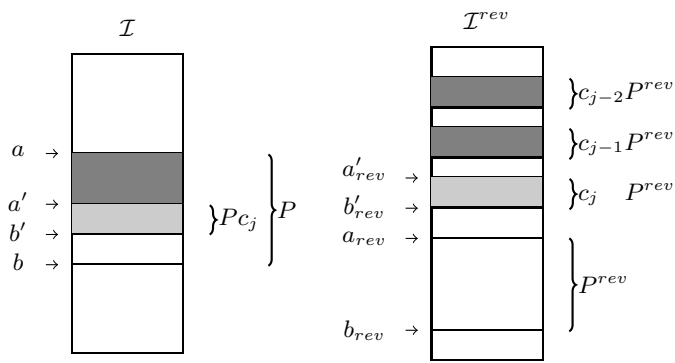

        \centering
        ~~~~~~~~~~~\subfloat{\biFMTwoA}~~~~~~~~~~~~
        \subfloat{\biFMTwoB}
    \caption{When conducting a forward search $P \Rightarrow Pc_j$ we need to determine the subinterval of the suffix array interval for $P$ which is depicted on the left. In order to determine the start, we can compute in $\Index^{rev}$ the size of the intervals for all characters smaller then $c_j$, depicted in dark gray on the right. The sum of all those sizes is exactly the needed offset from the beginning of the interval for $P$ in $\Index$.}\label{fig:2fmsearch}
\end{figure}

The implementation of the occurrence table $Occ$ is usually \emph{not} done by storing explicitly the values of the entire table. Instead of storing the entire $Occ : \Sigma \times \{1,\dots,n\} \rightarrow \{1,\dots,n\}$ one uses the more space-efficient \emph{constant time rank dictionary}: for every $c \in \Sigma$ a bit vector $B_c[1..n]$ is constructed such that $B_c[i] = 1$ if and only if $L[i] = c$. Thus the occurrence value equals the number of $1$'s in $B_c[1..i]$, i.e., $Occ(c, i) = rank(B_c, i)$. Jacobson \cite{jacobson1988succinct} showed that rank queries can be answered in constant time using only $\oh(n)$ additional space per bit vector. Since then many other constant time rank query data structures have been proposed. For an overview we refer the reader to \cite{Navarro_Fast_2012} containing a comparison of various implementations. Since we will make also use of this technique, we explain the most simple idea, namely the one for \emph{2-level rank dictionaries} in the following paragraph.

\subsection*{Constant time rank queries}
\label{subsec:constrankquery}

In order to store partial prefix sums, the technique uses two levels of lookup table, called \emph{blocks} and \emph{superblocks}. Given a bit vector $B$ of length $n$ we divide it into blocks of length $\ell$ and superblocks of length $\ell^2$ where
 \[\ell=\left\lfloor\frac{\log n}{2}\right\rfloor.\]
 For both, blocks and superblocks we allocate arrays $M$ and $M'$ of sizes $\left\lfloor\frac{n}{\ell}\right\rfloor$ and $\left\lfloor\frac{n}{\ell^2}\right\rfloor$ respectively (see Figure \ref{fig:2level} for an illustration).

For the $m$-th superblock we store the number of $1$'s from the beginning of $B$ to the end of the superblock in $M'[m]=rank(B, m \cdot \ell^2)$. As there are $\left\lfloor\frac{n}{\ell^2}\right\rfloor$ superblocks, $M'$ can be stored in $\Oh\left(\frac{n}{\ell^2} \cdot \log n\right)=\Oh\left(\frac{n}{\log n}\right)=\oh(n)$ bits.
For the $m$-th block we store the number of $1$'s from the beginning of the overlapping superblock to the end of the block in $M[m]=rank\big(B[1+k\ell..n],(m-k) \cdot \ell\big)$, where $k=\left\lfloor\frac{m-1}{\ell}\right\rfloor\ell$ is the total number of blocks in all superblocks left of the current superblock. $M$ has $\left\lfloor\frac{n}{\ell}\right\rfloor$ entries and can be stored in $\Oh\left(\frac{n}{\ell} \cdot \log \ell^2\right)=\Oh\left(\frac{ n\cdot\log\log n}{\log n}\right)=\oh(n)$ bits.

\begin{figure}[htbp]
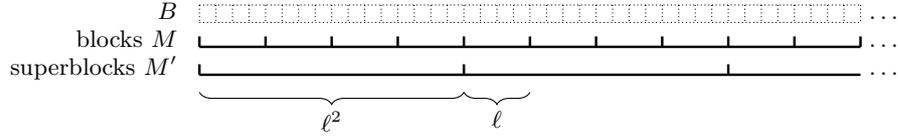

  \centering

  \hfill
  \rankdictGraphics
  \hfill
  \caption{2-level dictionary. Blocks and superblocks are allocated for each character (only one shown).}
    \label{fig:2level}
\end{figure}

Given a rank query $rank(B, i)$, one can now add the corresponding superblock and block values. But we still have to account for the $1$'s in the block covering position $i$ (in case $i$ is not at the end of a block). Let $P$ be a precomputed lookup table such that for each possible bit vector $V$ of length $\ell$ and $i\in[1..\ell]$ holds $P[V][i]=rank(V,i)$. $V$ has $2^\ell \cdot \ell$ entries of values at most $\ell$ and thus can be stored in \[\Oh\left(2^\ell \cdot \ell \cdot \log \ell \right)=\Oh\left(2^{\frac{\log n}{2}} \cdot \log n \cdot \log \log n \right)= \Oh\left(\sqrt{n} \cdot \log n \cdot \log\log n \right)=\oh(n)\]
bits.
We now decompose a rank query into 3 subqueries using the precomputed tables. For a position $i$ we determine the index $p=\left\lfloor\frac{i-1}{\ell}\right\rfloor$ of next block left of $i$ and the index $q=\left\lfloor\frac{p-1}{\ell}\right\rfloor$ of the next superblock left of block $p$. Then it holds:
$$rank(B,i)=M'[q]+M[p]+P\big[B[1+p\ell..(p+1)\ell]\big]\big[i-p\ell\big].$$

Since the text $T$ of length $n$ has to be addressed, we assume that a register has at least size $\left\lceil\log n\right\rceil$. Thus $B[1+p\ell..(p+1)\ell]$ fits into a single register and can be determined in $\Oh(1)$ time. Therefore a rank query can be answered in $\Oh(1)$ time. In practice the precomputed lookup table $P$ is replaced by a popcount operation on the CPU register and we have only two lookup operations.

One can now replace the occurrence table by this 2-level dictionary, i.e., by creating a bit vector for every $c \in \Sigma$ and setting it to $1$ if the character occurs in the BWT $L$. This results in $\Oh(\sigma\cdot n)+\oh(\sigma\cdot n)$ bits space requirement.

\end{document}